\documentclass[runningheads,a4paper]{llncs}
\usepackage{graphicx}
\usepackage[text={6in,10in},centering]{geometry}

\newcommand{\Nm}[1]{{||#1||}}
\usepackage{tikz}
\usepackage{tikz-cd}
\usetikzlibrary{arrows,automata}
\usetikzlibrary{shapes,snakes,backgrounds}
\usepackage[dvips]{epsfig}
\usepackage{amsfonts}
\usepackage{amssymb}
\usepackage{amscd}
\usepackage{amstext}
\usepackage{amsmath}
\usepackage{pstricks}
\usepackage{enumerate}
\usepackage{hyperref}
\usepackage{ulem}
\def\tcr#1{\textcolor{red}{#1}}  
  
\def\tcb#1{\textcolor{blue}{#1}}

\newtheorem{appli}[example]{Application}

\def\shuffle{\mathop{_{^{\sqcup\!\sqcup}}}}

\def\adots{\mathinner{\mkern2mu\raise1pt\hbox{.}
\mkern3mu\raise4pt\hbox{.}\mkern1mu\raise7pt\hbox{.}}}

\def\pointir{\unskip . -- \ignorespaces}

\def\up#1{\raise 1ex\hbox{\footnotesize#1}}

\def\supp{\mathop\mathrm{supp}\nolimits}
\def\span{\mathop\mathrm{span}\nolimits}

\def\ra{\rightarrow}

\def\lra{\leftrightarrow}
\def\Lra{\Longrightarrow}

\def\path{\rightsquigarrow}

\def\C{{\mathbb C}}
\def\R{{\mathbb R}}

\def\Q{{\mathbb Q}}

\def\S{{\Sigma}}
\def\Q{{\mathbb Q}}

\def\calF{{\mathcal F}}

\def\calX{{\mathcal X}}
\def\scal#1#2{\langle #1 | #2 \rangle}

\def\ncs#1#2{#1\langle\langle #2\rangle\rangle}

\def\ncp#1#2{#1\langle #2\rangle}

\def\dext#1#2{#1\{ \!\{ #2\} \!\}}

\def\dd{\mathbf{d}}

\newcommand\rsetminus{\mathbin{\mathpalette\rsetminusaux\relax}}
\newcommand\rsetminusaux[2]{\mspace{-4mu}
  \raisebox{\rsmraise{#1}\depth}{\rotatebox[origin=c]{-20}{$#1\smallsetminus$}}
 \mspace{-4mu}
}
\newcommand\rsmraise[1]{%
  \ifx#1\displaystyle .8\else
    \ifx#1\textstyle .8\else
      \ifx#1\scriptstyle .6\else
        .45%
      \fi
    \fi
  \fi}
\begingroup
\def\2#1{\ifnum#1<10 0\fi\the#1}
\xdef\isodayandtime{
{\2\day-\2\month-\the\year\space\2{\count0}:%
\2{\count2}}}
\endgroup

\def\QY{\poly{\Q}{Y}}

\def\QY_0{\Q\left\langle{Y_0}\right\rangle}

\usepackage{amsmath}
\usepackage{amsfonts}
\usepackage{amssymb}

\newcommand{\calA}{{\mathcal A}}

\newcommand{\calM}{{\mathcal M}}
\newcommand{\calH}{{\mathcal H}}
\newcommand{\calC}{{\mathcal C}}

\newcommand{\calD}{{\mathcal D}}

\newcommand{\N}{{\mathbb N}}

 \def\shuffle{\mathop{_{^{\sqcup\!\sqcup}}}}

\newcommand{\Li}{\operatorname{Li}}
\newcommand{\ad}{\operatorname{ad}}

\def\deg{\mathrm{deg}}

\newcommand{\poly}[2]{#1 \langle #2 \rangle}

\def\QY{\poly{\Q}{Y}}

\def\QY_0{\Q\left\langle{Y_0}\right\rangle}

\newcommand{\calT}{\mathcal{T}}

 \newcommand{\calB}{{\cal B}}

\def\path{\rightsquigarrow}

\def\bv{\mid}

\def\deg{\mathop\mathrm{deg}\nolimits}

\def\supp{\mathop\mathrm{supp}\nolimits}

\def\binom#1#2{{#1\choose#2}}

\usepackage{ulem}
\def\tcr#1{\textcolor{red}{#1}}  
  
\def\tcb#1{\textcolor{blue}{#1}} 
 
\def\pointir{\unskip . -- \ignorespaces} 
\def\ra{\rightarrow}
\def\lra{\longrightarrow}

\def\supp{\mathrm{supp}}
\def\scal#1#2{\langle #1\bv#2 \rangle}
\def\ncp#1#2{#1\langle #2\rangle}
\def\ncs#1#2{#1\langle \!\langle #2\rangle \!\rangle}
\usepackage[most]{tcolorbox}
\usetikzlibrary{patterns}
\pgfdeclarepatternformonly{mystrikeout}{\pgfqpoint{-1pt}{-1pt}}{\pgfqpoint{11pt}{11pt}}{\pgfqpoint{10pt}{10pt}}%
{
  \pgfsetlinewidth{0.4pt}
  \pgfpathmoveto{\pgfqpoint{0pt}{0pt}}
  \pgfpathlineto{\pgfqpoint{10.1pt}{10.1pt}}
  \pgfusepath{stroke}
}
\newtcolorbox{tcbstrikeout}{breakable,
 enhanced jigsaw,
 opacityback=0,
 parbox=false,
 boxrule=0mm,
 top=0mm,bottom=0pt,left=0pt,right=0pt,
 boxsep=0pt,
 frame hidden,
 finish={\fill[pattern=mystrikeout] (frame.north west) rectangle (frame.south east);}
}

\newtcbox{\lbox}[1][]{on line, sharp corners, colback=white, 
colframe=black, size=small, leftrule=0pt,#1} 
\newtcbox{\rbox}[1][]{on line, sharp corners, colback=white, 
colframe=black, size=small, rightrule=0pt,#1} 
\newtcbox{\tbox}[1][]{on line, sharp corners, colback=white, 
colframe=black, size=small, toprule=0pt,#1} 
\newtcbox{\bbox}[1][]{on line, sharp corners, colback=white, 
colframe=black, size=small, bottomrule=0pt,#1} 

\title{A localized version of the basic triangle theorem}
\subtitle{Lie-theoretic aspects of the basic triangle theorem}
\titlerunning{Lie-theoretic aspects of the BTT}
\author{G.H.E. Duchamp (a), N.P. Gargava (b), V. Hoang Ngoc Minh (c), P. Simonnet (d),}
\institute{(a) University Paris 13, Sorbonne Paris City, 93430 Villetaneuse, France,\\
(b) \'Ecole Polytechnique F\'ed\'erale de Lausanne, CH-1015, Lausanne, Switzerland,\\ 
(c) University of Lille, 1 Place D\'eliot, 59024 Lille, France,\\ 
(d) University of Corsica, 20250 Corte, France.}
\date{\isodayandtime}
\begin{document}
\maketitle
\begin{abstract}
In this paper, we examine some aspects of the BTT (basic triangle theorem), first published in 2011 (see \cite{Linz}).\\
In a first part, we review the interplay between integration of Non Commutative Differential equations and paths drawn on 
Magnus groups and some of their closed subgroups.\\    
In a second part, we provide a localized version of the BTT and aply it to prove the independence of hyperlogarithms over various function algebras. This version provides direct access to rings of scalars and avoids the recourse to fraction fields as that of meromorphic functions for instance.
\end{abstract}
%
\section{Introduction}
\textit{Iterated integrals (Lappo-Danielevskii), K.-T. Chen \cite{KTC,CFL} (path spaces, loops spaces, algebraic topology), Brown, Kontsevich}
\\
In a second step, we will use an analogue of the well-known closed subgroup theorem (also called Cartan theorem for finite dimensional Lie groups) which, in the Banach Lie context can be stated as follows.\\
Let $\calB$ be a Banach algebra (with unit $e$) and $G$ a closed subgroup of the open set $\calB^{\times}$. 
By a path ``drawn on $T$'' ($T\subset \calB$) is understood any function $\varphi:I\to T$ where $I$ is a open real interval.
\\ 
The first step is to establish what would be seen as the Lie algebra of $G$.\\ 
Let $L=L(G)$ be the set of tangent vectors of $G$ at the origin i.e.
\begin{equation}
L(G)=\{\gamma\,'(0)\ |\ \gamma:I\to G \mbox{ is differentiable at $0\ (=0_\R)$ and $\gamma(0)=e$}\}
\end{equation} 
\begin{proposition}[see \cite{MO1}] The set $L(G)$ has the following properties
\begin{enumerate}[1)]
\item\label{p1} If $u\in L(G)$, then, the one-parameter group $t\mapsto e^{t\cdot u}$ is drawn on $G$
\item\label{p2} $L(G)$ is a closed Lie subalgebra of $\calB$
\item\label{p3} Let $g\in \calB$ s.t. $\Nm{g-e}<1$, then 
$$
\log(g)\in L(G) \Longrightarrow g\in G
$$
\end{enumerate}
\end{proposition}  
\section{BTT theorem}
\subsection{Background}
Notations about alphabets and (noncommutative) series are standard and can be found 
in \cite{berstel}. \\
\textit{Set of variables, series, Dirac-Sch\"utzenberger duality, Magnus and Hausdorff groups. Series with constant and variable coefficients. Differential rings and algebras.}
\subsection{For the Magnus group}
We can always consider a series with variable coefficients $S\in \ncs{\calH(\Omega)}{\calX}$ as a \textit{function} i.e. 
with, for all $z\in \Omega$ 
\begin{equation}
S(z):=\sum_{w\in \calX^*}\scal{S}{w}(z)\, w
\end{equation}
we get an embedding $\ncs{\calH(\Omega)}{\calX}\hookrightarrow \ncs{\C}{\calX}^{\Omega}$. With this point of view in head, we can always consider series $S\in \ncs{\calH(\Omega)}{\calX}$ such that $\scal{S}{1_{\calX^*}}=1_{\calH(\Omega)}$ as (holomorphic) paths \textit{drawn on the Magnus group.}   
The 
Non commutative differential equations with left multiplier can be expressed in the context of general 
differential algebras. 
\begin{theorem}[See Th 1 in \cite{Linz}]\label{ind_lin} Let $(\calA,d)$ be a $k$-commutative associative differential algebra with unit ($ker(d)=k$, a field) and $\calC$ be a differential subfield of $\calA$ 
(i.e. $d(\calC)\subset \calC$ and $k\subset \calC$). 
Let $X$ be some alphabet (i.e. some set) and we define 
$\mathbf{d}:\ncs{\calA}{X}\rightarrow\ncs{\calA}{X} $ to be the map given by 
$\scal{\dd(S)}{w} = d(\scal{S}{w})$.
We suppose that $S\in \ncs{\calA}{X}$ is a solution of the differential equation
\begin{equation}
  \dd(S)=MS\ ;\ \scal{S}{1_{X^{*}}}=1_{\calA}
\end{equation}
where the multiplier $M$ is a homogeneous series (a polynomial in the case of finite $X$) of degree $1$,
i.e.
\begin{equation}
M=\sum_{x\in X}u_x x\in \ncs{\calC}{X}\ .
\end{equation}
The following conditions are equivalent :
\begin{enumerate}[i)]
	\item The family $(\scal{S}{w})_{w\in X^*}$ of coefficients of $S$ is free over $\calC$.
  \item The family of coefficients $(\scal{S}{y})_{y\in X\cup \{1_{X^*}\}}$ is free over $\calC$.
	\item The family $(u_x)_{x\in X}$ is such that, for $f\in \calC$ and $\alpha\in k^{(X)}$ (i.e. 
	$supp(\alpha)$ is finite)
\begin{equation}\label{prim_ind}
d(f)=\sum_{x\in X} \alpha_x u_x \Longrightarrow (\forall x\in X)(\alpha_x=0)\ .
\end{equation}
	\item The family $(u_x)_{x\in X}$ is free over $k$ and
\begin{equation}
	d(\calC)\cap\span_k\Big((u_x)_{x\in X}\Big)=\{0\}\ .
\end{equation}
\end{enumerate}
\end{theorem}

\begin{proof} For convenience of the reader, we enclose here the demonstration given in Th 1 \cite{Linz}.
\\
(i)$\Longrightarrow$(ii) Obvious.\\
(ii)$\Longrightarrow$(iii)\\
Suppose that the family $(\scal{S}{y})_{y\in X\cup \{1_{X^*}\}}$ (coefficients taken at letters and the empty word) of coefficients of $S$ were free over $\calC$ and let us consider the relation as in \eqref{prim_ind}
\begin{equation}
d(f)=\sum_{x\in X} \alpha_x u_x\ .
\end{equation}
We form the polynomial $P=-f1_{X^*}+\sum_{x\in X}\alpha_x x$. One has $\dd(P)=-d(f)1_{X^*}$ and
\begin{equation}
	d(\scal{S}{P})=\scal{\dd(S)}{P}+\scal{S}{\dd(P)}=\scal{MS}{P}-d(f)\scal{S}{1_{X^*}}=
	(\sum_{x\in X} \alpha_x u_x)-d(f)=0
\end{equation}
whence  $\scal{S}{P}$ must be a constant, say $\lambda\in k$. For $Q=P-\lambda.1_{X^*}$, we have
$$\supp(Q)\subset X\cup \{1_{X^*}\}\ \textrm{and}\ \scal{S}{Q}=\scal{S}{P}-\lambda\scal{S}{1_{X^*}}=\scal{S}{P}-\lambda=0
\ .
$$
This, in view of (ii), implies that $Q=0$ and, as $Q=-(f+\lambda)1_{X^*}+\sum_{x\in X}\alpha_x x$, one has, in particular, $supp(\alpha)=\emptyset$ (and, as a byproduct, $f=-\lambda$ which is indeed the only possibility for the L.H.S. of \eqref{prim_ind} to occur).\\
(iii)$\Longleftrightarrow$(iv)\\
Obvious, (iv) being a geometric reformulation of (iii).\\
(iii)$\Longrightarrow$(i)\\
Let $\mathcal{K}$ be the kernel of $P\mapsto \scal{S}{P}$ (linear $\ncp{\calC}{X}\ra \calA$) i.e.
\begin{equation}
	\mathcal{K}=\{P\in \ncp{\calC}{X}| \scal{S}{P}=0\}\ .
\end{equation}
If $\mathcal{K}=\{0\}$, we are done. Otherwise, let us adopt the following strategy.\\
First, we order $X$ by some well-ordering $<$ (\cite{B_Sets} III.2.1) and $X^*$ by the graded lexicographic ordering $\prec$ defined as follows
\begin{equation}
	u\prec v \Longleftrightarrow |u|<|v|\ \textrm{or}\ (u=pxs_1\ ,\ v=pys_2\ \textrm{and}\ x<y).
\end{equation}
It is easy to check that $X^*$ is also a well-ordered by $\prec$. For each nonzero polynomial $P$, we denote by $lead(P)$ its leading monomial; i.e. the greatest element of its support $\supp(P)$ (for $\prec$).\\
Now, as $\mathcal{R}=\mathcal{K}\rsetminus \{0\}$ is not empty, let $w_0$ be the minimal element of $lead(\mathcal{R})$ and choose a $P\in \mathcal{R}$ such that $lead(P)=w_0$. We write
\begin{equation}
	P=fw_0+\sum_{u\prec w_0}\scal{P}{u}u\ ;\ f\in \calC\rsetminus\{0\}\ .
\end{equation}
The polynomial $Q=\frac{1}{f}P$ is also in $\mathcal{R}$ with the same leading monomial, but the leading coefficient is now $1$; and so $Q$ is given by
\begin{equation}
	Q=w_0+\sum_{u\prec w_0}\scal{Q}{u}u\ .	
\end{equation}

Differentiating $\scal{S}{Q}=0$, one gets
\begin{eqnarray}
&&	0=\scal{\dd(S)}{Q}+\scal{S}{\dd(Q)}=\scal{MS}{Q}+\scal{S}{\dd(Q)}=\cr
&&	\scal{S}{M^\dagger Q}+\scal{S}{\dd(Q)}=\scal{S}{M^\dagger Q+\dd(Q)}
\end{eqnarray}
with
\begin{equation}
M^\dagger Q+\dd(Q)=\sum_{x\in X} u_x (x^\dagger Q)+\sum_{u\prec w_0}d(\scal{Q}{u})u\in \ncp{\calC}{X}	\ .
\end{equation}
It is impossible that $M^\dagger Q+\dd(Q)\in \mathcal{R}$ because it would be of leading monomial strictly less than $w_0$, hence $M^\dagger Q+\dd(Q)=0$. This is equivalent to the recursion
\begin{equation}
	d(\scal{Q}{u})=-\sum_{x\in X} u_x \scal{Q}{xu}\ ;\ \textrm{for}\ x\in X\ ,\ v\in X^*   .
\end{equation}
From this last relation, we deduce that $\scal{Q}{w}\in k$ for every $w$ of length $deg(Q)$ and,
because $\scal{S}{1_{X^{*}}}=1_\calA$, one must have $deg(Q)>0$. Then, we write $w_0=x_0v$ and compute the coefficient at $v$
\begin{equation}
d(\scal{Q}{v})=-\sum_{x\in X} u_x \scal{Q}{xv}=\sum_{x\in X} \alpha_x u_x
\end{equation}
with coefficients $\alpha_x=-\scal{Q}{xv}\in k$ as $|xv|=\deg(Q)$ for all $x\in X$.
Condition \eqref{prim_ind} implies that all coefficients $\scal{Q}{xv}$ are zero; in particular, as
$\scal{Q}{x_0v}=1$, we get a contradiction. This proves that $\mathcal{K}=\{0\}$.\\
\end{proof}
\section{Localization}
We will now establish the following extension of Theorem 1 in \cite{Linz}.
Let $(\calA,d)$ be a $k$-commutative associative differential algebra with unit ($ker(d)=k$, a field). We consider a solution of the differential equation
\begin{equation}\label{NCDE_abs}
  \dd(S)=MS\ ;\ \scal{S}{1_{X^{*}}}=1_{\calA}
\end{equation}
where the multiplier $M$ is a homogeneous series (a polynomial in the case of finite $X$) 
of degree $1$, i.e.
\begin{equation}
M=\sum_{x\in X}u_x x\in \ncs{\calA}{X}\ .
\end{equation} 
\begin{proposition}[Thm1 in \cite{Linz}, Localized form]\label{ind_lin2} Let $(\calA,d)$ be a commutative associative differential ring ($ker(d)=k$ being a field) and $\calC$ be a differential subring (i.e. $d(\calC)\subset \calC$) of $\calA$ which is an integral domain containing the field of constants.
\\
We suppose that, for all $x\in X$, $u_x\in \calC$ and that $S\in \ncs{\calA}{X}$ is a solution of the differential equation \eqref{NCDE_abs} and that $(u_x)_{x\in X}\in \calC^{X}$.\\
The following conditions are equivalent :
\begin{enumerate}[i)]
\item The family $(\scal{S}{w})_{w\in X^*}$ of coefficients of $S$ is free over $\calC$.
\item The family of coefficients $(\scal{S}{y})_{y\in X\cup \{1_{X^*}\}}$ is free over $\calC$.
\item[iii')] For all $f_1,f_2\in \calC,\ f_2\not=0$ and $\alpha\in k^{(X)}$, we have the property 
\begin{equation}
W(f_1,f_2)=f_2^2(\sum_{x\in X} \alpha_x u_x) \Longrightarrow (\forall x\in X)(\alpha_x=0)\ .
\end{equation}
where $W(f_1,f_2)$, the wronskian, stands for $d(f_1)f_2-f_1d(f_2)$. 
\end{enumerate}
\end{proposition} 
\begin{proof} (i.$\Lra$ ii.) being trivial, remains to prove (ii.$\Lra$ iii'.) and (iii'.$\Lra$ i.). To this end, we localize the situation w.r.t. the multiplicative subset $\calC^\times:= \calC\rsetminus \{0\}$ as can be seen in the following commutative cube
\begin{center}
\begin{equation}
\begin{tikzcd}[row sep=2.5em]
\calC \arrow[rr,hook,"\varphi_{\calC}"] \arrow[dr,hook,swap,"j"] \arrow[dd,swap,"d"] 
&&
Fr(\calC) \arrow[dd,"d_{frac}" near start] \arrow[dr,hook,"j_{frac}"] 
\\
& \calA \arrow[rr,crossing over,"\varphi_{\calA}" near start] &&
\calA[(\calC^\times)^{-1}] \arrow[dd,"d_{frac}"] \\
\calC \arrow[rr,hook,"\varphi_{\calC}" near end] \arrow[dr,hook,swap,"j"] 
&& Fr(\calC) \arrow[dr,hook,swap,"j_{frac}"] \\
& \calA \arrow[rr,"\varphi_{\calA}"] \arrow[uu,<-,crossing over,"d" near end]&& \calA[(\calC^\times)^{-1}]
\end{tikzcd}
\end{equation}
\end{center}
We give here a detailed demonstration of the commutation which provides, in passing, the labelling of the arrows.\\
\textbf{Left face}\pointir Comes from the fact that $d(\calC)\subset \calC$, $j$ being the canonical embedding.\\
\textbf{Upper and lower faces}\pointir We first construct the localization\\ $\varphi_{\calA}: \calA\lra\calA[(\calC^\times)^{-1}]$ w.r.t. the multiplicative subset $\calC^\times\subset \calA\rsetminus \{0\}$ (recall that $\calC$ has no zero divisor). Now, from standard theorems (see \cite{CA}, ch2 par. 2 remark 3 after Def. 2, for instance), we have 
\begin{equation}\label{ker1}
\ker(\varphi_{\calA})=\{u\in \calA|(\exists v\in \calC^\times)(uv=0)\}
\end{equation}
For every intermediate ring $\calC\subset \calB\subset\calA$, we remark that the composittion 
$$
\calB\hookrightarrow \calA\stackrel{\varphi_{\calA}}{\lra} \calA[(\calC^\times)^{-1}]
$$
realises the ring of fractions $\calB[(\calC^\times)^{-1}]$ which can be identified with the subalgebra generated by $\varphi_{\calA}(\calB)$ and the set of inverses $\varphi_{\calA}(\calC^\times)^{-1}$. Applying this to $\calC$, and remarking that $\calC[(\calC^\times)^{-1}]\simeq Fr(\calC)$, we get the embedding $j_{frac}$ and the commutation of upper and lower faces.\\  
\textbf{Front and rear faces}\pointir From standard constructions (see e.g. the book \cite{VdP}), there exists a unique $d_{frac}\in \mathfrak{Der}(\calA[\calC^\times)^{-1}])$ such that these faces commute.\\ 
\textbf{Right face}\pointir Commutation comes from the fact that $d_{frac}j_{frac}$ and 
$j_{frac}d_{frac}$ coincide on $\varphi_{\calC}(\calC)$ hence on  $\varphi_{\calC}(\calC^\times)$ and on their inverses. Therefore on all $Fr(\calC)$.\\    
From the constructions it follows that the arrows (derivations, morphisms) are arrows of $k$-algebras.

\smallskip
Now, we set 
\begin{enumerate}
\item $\bar{S}=\sum_{w\in X^*}\varphi_{\calA}(\scal{S}{w})w\in \ncs{\calA[\calC^\times)^{-1}]}{X}$
\item $\bar{M}=\sum_{x\in X}\varphi_{\calC}(u_x)\,x\in \ncs{\calA[\calC^\times)^{-1}]}{X}$
\end{enumerate}
it is clear, from the commutations, that $(\ncs{\calA[\calC^\times)^{-1}]}{X},\dd_{frac})$ where 
$\dd_{frac}$ is the extension of $d_{frac}$ to the series, is a differential algebra and that 
\begin{equation}
\dd_{frac}(\bar{S})=\bar{M}\bar{S}\ ;\ \scal{\bar{S}}{1}=1
\end{equation}
we are now in the position to resume the proofs of (ii.$\Lra$ iii'.) and (iii'.$\Lra$ i.).\\
ii.$\Lra$ iii'.) Supposing (ii), we remark that the family of coefficients 
$$
(\scal{\bar{S}}{y})_{y\in X\cup \{1_{X^*}\}}
$$ 
is free over $\calC$\footnote{As $\varphi_{\calC}$ is injective on $\calC$ we identify 
$\varphi_{\calC}(\calC)$ and $\calC$, this can be unfolded on request, of course.}.
Indeed, let us suppose a relation 
\begin{equation}\label{lin_rel1}
\sum_{y\in X\cup \{1_{X^*}\}} g_y\,\scal{\bar{S}}{y}=0\mbox{ with } (g_y)_{y\in X\cup \{1_{X^*}\}}
\in \calC^{(X\cup \{1_{X^*}\})}
\end{equation} 
this relation is equivalent to 
\begin{equation}\label{lin_rel2}
\varphi_{\calA}(\sum_{y\in X\cup \{1_{X^*}\}} g_y\,\scal{S}{y})=0
\end{equation} 
which, in view of \eqref{ker1}, amounts to the existence of $v\in \calC^\times$ such that 
\begin{equation}\label{lin_rel3}
0=v(\sum_{y\in X\cup \{1_{X^*}\}} g_y\,\scal{S}{y})=\sum_{y\in X\cup \{1_{X^*}\}} vg_y\,\scal{S}{y}
\end{equation} 
which implies $(\forall y\in X\cup \{1_{X^*}\})(vg_y=0)$ but, $\calC$ being without zero divisor, one gets
\begin{equation}
(\forall y\in X\cup \{1_{X^*}\})(g_y=0)
\end{equation}  
which proves the claim. This implies in particular, by chasing denominators, that the family of coefficients 
$$
(\scal{\bar{S}}{y})_{y\in X\cup \{1_{X^*}\}}
$$ 
is free over $Fr(\calC)$. This also implies\footnote{And indeed is equivalent under the assumption of (ii).} that $\varphi_{\calA}$ is injective on  
\begin{equation}\label{span1}
span_{\calC}(\scal{S}{y})_{y\in X\cup \{1_{X^*}\}}
\end{equation} 
To finish the proof that (ii.$\Lra$ iii'.), let us choose $f_1,f_2\in \calC$ with $f_2\not=0$ and set some relation which reads
\begin{equation}\label{lin_rel4}
W(f_1,f_2)=f_2^2(\sum_{x\in X} \alpha_x u_x)
\end{equation} 
with $\alpha\in k^{(X)}$, then 
\begin{equation}\label{lin_rel5}
(\sum_{x\in X} \alpha_x u_x)=\frac{W(f_1,f_2)}{f_2^2}=d_{frac}(\frac{f_1}{f_2})
\end{equation} 
but, in view of Th1 in \cite{Linz} applied to the differential field $Fr(\calC)$, we get $\alpha\equiv 0$.
\\
(iii'.$\Lra$ i.) The series $\bar{S}$ satisfies 
\begin{equation}\label{NCDE_abs_frac}
	\dd(\bar{S})=\bar{M}\bar{S}\ ;\ \scal{\bar{S}}{1_{X^*}}=1_{\calA[\calC^\times)^{-1}]}=1_{Fr(\calC)}
\end{equation}
and remarking that
\begin{enumerate}
\item all $f$ in the differential field $Fr(\calC)$ can be expressed as $f=\frac{f_1}{{f_2}}$ 
\item condition (iii') for $(S,\calA,\calC,d,X)$ implies condition (iii) for 
$(\bar{S},\calA[\calC^\times)^{-1}],Fr(\calC),d_{frac},X)$\footnote{Once again we identify, with no loss, 
$k\subset \calC$, the latter being idetified with its image through $\varphi_{\calA[\calC^\times)^{-1}]}$.} which, in turn, implies the $Fr(\calC)$-freeness of $(\scal{\bar{S}}{w})_{w\in X^*}$ hence its $\calC$-freeness and, by inverse image\footnote{If the image (through a $A$-linear arrow) of a family is $A$-free then the family itself is $A$-free.} the $\calC$-freeness of $(\scal{S}{w})_{w\in X^*}$.
\end{enumerate}
\end{proof}
\begin{remark} i) It seems reasonable to think that the whole commutation of the cube could be understood by natural transformations within an appropriate category. If yes, this will be inserted in a forthcoming version.\\
ii) In fact, in the localized form and with $\calC$ not a differential field, $(iii)$ is strictly weaker 
than $(iii')$, as shows the following family of counterexamples 
\begin{enumerate}
\item $\Omega =\C\rsetminus (]-\infty,0])$
\item $X=\{x_0\}$, $u_0=z^\beta,\ \beta\notin\Q$
\item $\calC_0=\dext{\C}{z^{\beta}}=\C.1_{\Omega}\oplus span_{\C}\{z^{(k+1)\beta-l}\}_{k,l\ge0}$
\item $S=1_{\Omega}+(\sum_{n\geq 1}\frac{z^{n(\beta+1)}}{(\beta+1)^nn!})$
\end{enumerate}
\end{remark}
\begin{appli} As a result of the theory of domains (see \cite{KSOP}), the conc-characters 
$(\alpha x_0)^*,(\beta x_1)^*$ are in the domain of $\Li_{\bullet}$ (see \cite{KSOP} for details), then due to the fact that 
$\calH(\Omega)$ is nuclear, their shuffle 
$(\alpha x_0)^*\shuffle (\beta x_1)^*=(\alpha x_0+\beta x_1)^*$ is also in $Dom(\Li_{\bullet})$. Let us compute 
\begin{equation}
\Li(\alpha x_0)^*\shuffle (\beta x_1)^*))=\Li(\alpha x_0)^*\Li(\beta x_1)^*)=z^{\alpha}(1-z)^{\beta}
\end{equation}
Now, for a family of functions $\calF=(f_{i})_{i\in I}$, let us note $\C\{f_i\}_{i\in I}$ the algebra generated by $\calF$ within $\calH(\Omega)$ and then set $\calC_{\C}:=\C\{z^{\alpha}(1-z)^{\beta}\}_{\alpha,\beta\in \C}$. We, at once, remark that, as $\calM=\{z^{\alpha}(1-z)^{\beta}\}_{\alpha,\beta\in \C}$ is a monoid, 
$$
\calC_{\C}=span_{\C}(z^{\alpha}(1-z)^{\beta})_{\alpha,\beta\in \C}=\C[z^{\alpha}(1-z)^{-\beta}]_{\alpha,\beta\in \C}
$$ 
as well.\\
In this aplication, we give a detailed proof that the family $(\Li_w)_{w\in X^*}$ is $\calC_{\C}$-linearly independent.\\
Let us suppose $P_i\in\calC_{\C},i=1\ldots 3$\footnote{i.e. elements of the algebra of the monoid $\calM=\{z^{\alpha}(1-z)^{\beta}\}_{\alpha,\beta\in \C}$  
$\{z^{\alpha}(1-z)^{\beta}\}_{\alpha,\beta\in \C}$} such that 
$$
P_1(z)+P_2(z)\log(z)+P_3(z)(\log(\frac{1}{1-z}))=0_{\Omega}
$$
We first prove that $P_2=\sum_{i\in F}c_i z^{\alpha_i}(1-z)^{\beta_i}$ is zero using the deck transformation $D_0$ of index 
one around zero.\\
One has 
$D_0^n(\sum_{i\in F}c_i z^{\alpha_i}(1-z)^{\beta_i})=\sum_{i\in F}c_i z^{\alpha_i}(1-z)^{\beta_i}e^{2i\pi.n\alpha_i}$, the same calculation holds for all $P_i$  which proves that all $D_0^n(P_i)$ are bounded. But one has 
$D_0^n(\log(z))=\log(z)+2i\pi.n$ and then
\vspace{-3mm} 
\begin{eqnarray*}
&&D_0^n(P_1(z)+P_2(z)\log(z)+P_3(z)(\log(\frac{1}{1-z})))=\cr
&&D_0^n(P_1(z))+D_0^n(P_2(z))(\log(z)+2i\pi.n)+
D_0^n(P_3(z))\log(\frac{1}{1-z})=0
\end{eqnarray*}
It suffices to build a sequence of integers $n_j\to +\infty$ such that $\lim_{j\to\infty}D_0^{n_j}(P_2(z))=P_2(z)$ which is a consequence of the following lemma. 
\begin{lemma}{}
Let us consider a homomorphism $\varphi : \N\to G$ where $G$ is a compact (Hausdorff) group, then it exists $u_j\to +\infty$ such that 
$$
\lim_{j\to \infty}\varphi(u_j)=e
$$
\end{lemma}
\begin{proof} First of all, due to the compactness of $G$, the sequence $\varphi(n)$ admits a subsequence $\varphi(n_k)$ convergent to some $\ell\in G$. Now one can refine the sequence as $n_{k_j}$ such that 
$$
0<n_{k_1}-n_{k_0}<\ldots <n_{k_{j+1}}-n_{k_j}<n_{k_{j+2}}-n_{k_{j+1}}<\ldots
$$
With $u_j=n_{k_{j+1}}-n_{k_j}$ one has $\lim_{j\to \infty}\varphi(u_j)=e$.\\
\textbf{End of the proof} One applies the lemma to the morphism
\vspace{-2mm}
$$
n\mapsto (e^{2i\pi.n\alpha_i})_{i\in F}\in \mathbb{U}^F
$$  
\end{proof}    
\end{appli}
\section{Closed subgroup property and algebraic independance.}
\section{Appendix.}
\subsection{Closed subgroup (Cartan) theorem in Banach algebras}
\tcr{This section is meant to be withdrawn afterwards\footnote{Following the advise of Gauss ``no
self-respecting architect leaves the scaffolding in place after completing the building''.} and integrated within the introduction.}\\
Let $\calB$ be a Banach algebra (with unit $e$) and $G$ a closed subgroup of the open set $\calB^{\times}$. 
By a path ``drawn on $T$'' ($T\subset \calB$) is understood any function $\varphi:I\to T$ where $I$ is a open real interval.
\\ 
The first step is to establish what would be seen as the Lie algebra of $G$.\\ 
Let $L=L(G)$ be the set of tangent vectors of $G$ at the origin i.e.
\begin{equation}
L(G)=\{\gamma\,'(0)\ |\ \gamma:I\to G \mbox{ is differentiable at $0\ (=0_\R)$ and $\gamma(0)=e$}\}
\end{equation} 
\begin{proposition} The set $L(G)$ has the following properties
\begin{enumerate}[1)]
\item\label{p1} If $u\in L(G)$, then, the one-parameter group $t\mapsto e^{t\cdot u}$ is drawn on $G$
\item\label{p2} $L(G)$ is a closed Lie subalgebra of $\calB$
\item\label{p3} Let $g\in \calB$ s.t. $\Nm{g-e}<1$, then 
$$
\log(g)\in L(G) \Longrightarrow g\in G
$$
\end{enumerate}
\end{proposition}  
\begin{proof} \ref{p1}) 
Let $\gamma$ be such a tangent path (differentiable at $0$ and s.t. $\gamma(0)=e$), then one can write 
\begin{equation}
\gamma(t)=e+t\cdot \gamma\,'(0)+t.\epsilon_1(t)\mbox{ with } \lim_{t\to 0}\epsilon_1(t)=0_\calB
\end{equation}
and then, $t\in [-1,1]$ being fixed, 
\begin{equation}
\gamma(\dfrac{t}{n})=e+\dfrac{t}{n}\cdot \gamma\,'(0)+\dfrac{t}{n}\cdot \epsilon_2(n)
\mbox{ with } \lim_{n\to \infty}\epsilon_2(n)=0_\calB
\end{equation}
now, for $\Nm{h}<1$, one has (in $\calB$) $\log(e+h)=h+\Nm{h}\cdot \epsilon_3(h)$ so, with\\ 
$h=\dfrac{t}{n}\cdot \gamma\,'(0)+\dfrac{t}{n}\cdot \epsilon_2(n)$, we get 
\begin{equation}
n\log(\gamma(\dfrac{t}{n}))=t\cdot (\gamma\,'(0)+\epsilon_2(n))+\Nm{t\cdot (\gamma\,'(0)+\epsilon_2(n))}\epsilon_3(h)
\end{equation}
hence $\lim_{n\to \infty}n\log(\gamma(\dfrac{t}{n}))=t\cdot \gamma\,'(0)$ and then 
\begin{equation}
\lim_{n\to \infty}\gamma(\dfrac{t}{n})^n=\lim_{n\to \infty}e^{n\cdot\log(\gamma(\dfrac{t}{n}))}=e^{t\cdot \gamma\,'(0)}
\end{equation} 
as $\gamma$ is drawn on $G$, each $\gamma(\dfrac{t}{n})^n$ belongs to $G$ and it is the same for the limit ($G$ is closed), then $e^{t\cdot \gamma\,'(0)}\in G$ for all $t\in [-1,1]$. For general $t\in \R$, one just has to use the archimedean property that, for some $N\in \N_{\ge1}$, $\dfrac{t}{N}\in [-1,1]$ and remark that 
$e^{t\cdot \gamma\,'(0)}=(e^{\dfrac{t}{N}\cdot \gamma\,'(0)})^N$.  
\\
\ref{p2}) To prove that $L(G)$ is a Lie subalgebra of $\calB$ it suffices to provide suitable paths. Let $u,v\in L(G)$, we have

\smallskip 
\begin{center}
\begin{tabular}{c|c}
Path & Tangent vector (at zero)\\
\hline
\hline
$e^{t\cdot u}e^{t\cdot v}$ & $u+v$\\
\hline
$e^{\alpha t\cdot u}$ & $\alpha u$\\
\hline
Let $\gamma_1(t)=e^{r(t)\cdot u}e^{r(t)\cdot v}e^{-r(t)\cdot u}e^{-r(t)\cdot v}$ &\\
with $r(t)=\sqrt{2\cdot t}$ for $t\ge0$  & $[u,v]$\\
$\gamma(t)=\gamma_1(t)$ for $t\ge0$ and\\
$=(\gamma_1(-t))^{-1}$ for $t\le0$\\
\hline
\end{tabular}
\end{center}

\smallskip 
Remains to show that $L(G)$ is closed. To see this, let us consider a sequence $(u_n)_{n\in \N}$ in $L(G)$ which converges 
to $u$. For every fixed $t\in \R$, $\lim_{n\to \infty}e^{t\cdot u_n}=e^{t\cdot u}$ because $exp:\ \calB\to\calB$ is continuous, then $t\to e^{t\cdot u}$ is drawn on $G$ and $u\in L(G)$.   

\ref{p3}) Set $u=\log(g)$. Now, as $u\in L(G)$, the one-parameter group 
$t\to e^{t\cdot u}$ is drawn on $G$ and $g=(e^{t\cdot u})\bigr\rvert_{t=1}$ (in 
the neighourhood $\Nm{g-e}<1$, we have $\exp(\log(g))=g$).
\end{proof} 
Now, we have an analogue of Cartan's theorem in the context of Banach algebras 
\begin{theorem}\label{cartan} Let $G\subset \calB$ be a closed subgroup of $\calB^{\times}$ and $L(G)$ as above.
Let $I\subset \R$ be a non-void open interval and $M: I\to L(G)$ to be a continuous path drawn on $L(G)$. 
Let $t_0\in I,\ g_0\in G$. Then 
\begin{enumerate}[1)]
\item\label{p2.1} The system 
\begin{equation}
\Sigma(t_0,M,g_0)\qquad \left\{
\begin{array}{ll}
\dfrac{d}{dt}(S(t))=M(t).S(t) &\mbox{   (NCDE)}\cr
S(t_0)=g_0 & \mbox{    (Init. Cond.)}
\end{array}
\right.
\end{equation}
admits a unique solution $S:\ I\to \calB$.
\item\label{p2.2} This solution $S$ is a path drawn on $G$.
\end{enumerate}
\end{theorem}
\begin{proof}\ref{p2.1}-\ref{p2.2}) I sketch the proof below 
\begin{enumerate}[1)]
\item Let $J$ be an open real interval, $t_0\in J$ and $m\in C^0(J,L(G))$. 
In order to paste them together, we call ``local solution'' (of $\Sigma(t_0,m,g)$), 
a $C^1$ map $J\to G$ fulfilling the following system
\begin{equation}
\Sigma(t_0,m,g)\qquad \left\{
\begin{array}{ll}
\dfrac{d}{dt}(T(t))=m(t).T(t) &\mbox{   (NCDE)}\cr
T(t_0)=g & \mbox{    (Init. Cond.)}
\end{array}
\right.
\end{equation}
\item We first prove that there exists a local solution to any system $\Sigma(0,m,e)$ for\\ 
$t_0=0,\ g=e$. As 
$m$ is continuous, there is an open interval $J$, containing zero, in which $|t|<1$ and is such that $\Nm{m_1(t)}\leq B<1$. 
In these conditions, Picard's process 
\begin{equation}
\left\{
\begin{array}{l}
T_0=e\cr
T_{n+1}(t)=\int_{0}^t m(s).T_{n}(s)ds
\end{array}
\right.
\end{equation}
converges absolutely (in $C^0(J,\calB)$) to a function $T\in C^0(J,\calB)$ such that  
\begin{equation}
T(t)=e+\int_{0}^t m(s).T(s)ds
\end{equation}
this proves that, in fact, $T\in C^1(J,\calB)$. Remains to prove that $T$ is drawn on $G$. 
    
\item If $t$ and $B$ are sufficiently small \tcr{est1}, we have $\Nm{T(t)-e}<1$ and can compute $\Omega(t)=\log(T)$, which 
is $C^1$ and, by Magnus expansion (see below \ref{ME}), satisfies 
\begin{equation}
\Omega'(t)=\dfrac{ad_{\Omega}}{e^{ad_{\Omega}}-Id_{\calB}}[\Omega]
\end{equation} 
where, the symbol $\dfrac{ad_{\Omega}}{e^{ad_{\Omega}}-Id_{\calB}}$ denotes the substitution of $ad_{\Omega}$ in the series 
$\sum_{n\ge0}B_n\dfrac{z^n}{n!}$ ($B_n$ being the Bernouilli numbers) \tcr{est2}. But we know that $\Omega(t)$ is the limit of the following process
\begin{equation}
\left\{
\begin{array}{l}
\Omega_1(t)=\int_{0}^t m(s)ds\cr
\Omega_{n}(t)=\sum_{j=1}^{n-1}\dfrac{B_j}{j!}\sum_{k_1+\cdots +k_j=n-1\atop k_i\ge1}\int_{0}^t 
ad_{\Omega_{k_1}}\cdots ad_{\Omega_{k_j}}[m(s)]ds
\end{array}
\right.
\end{equation}
each $\Omega_n$ is drawn on $L(G)$ as shows the preceding recursion. Hence $T(t)=e^{\Omega(t)}$ is drawn on $G$. 
\item We can now shift the situation in order to compute a local solution of any system $\Sigma(t_0,m,g)$ as follows 
(given $J$ an open real interval, $t_0\in J$ and $m\in C^0(J,L(G))$ 
\begin{itemize}
\item Find a local solution $R$ of $\Sigma(0,m_1,e)$ with $m_1(t)=m(t+t_0)-m(t_0)$. For it 
\begin{equation}
R'(t)=(m(t+t_0)-m(t_0))R(t)\ ;\ T(0)=e
\end{equation}
\item\label{local1} Define $T(t):=R(t)e^{t\cdot m(t_0)}g$, one has
\begin{equation}
T'(t)=m(t)T(t)\ ;\ T(t_0)=g
\end{equation}
\end{itemize}
\item We now return to our original system $\Sigma(t_0,M,g_0)$. By the previous item \eqref{local1} we know that it admtis at least a local solution $(J,S)$. We remark also that if we have two local solutions $(J_1,S_1),\ (J_2,S_2)$ they coincide on 
$J_1\cap J_2$, thus th union of all graphs on local solutions of $\Sigma(t_0,M,g_0)$ is functional and is the maximal solution $(J_{max},S_{max})$ of $\Sigma(t_0,M,g_0)$. Now, if we had $b_m=sup(J_{max})<sup(I)$, we could consider the system 
$\Sigma(b_m,M,e)$ and a local solution $T$ of it on some $]a,b[$ with $a<b_m<b$, now taking some intermediate point $t_1$ within $]a,b_m[$, we observe that $TS_{max} T^{-1}(t_1)S_{max}(t_1)$ and $S_{max}$ coincide on $]a,b_m[$ and the union of their graphs would be a strict extension of $S_{max}$. A contradiction, then $b_m=sup(I)$. A similar reasoning proves that 
$a_m=inf(I)$ and then theorem \ref{cartan} is proved. 
\end{enumerate}
\end{proof}
\subsection{About Magnus expansion and Poincar\'e-Hausdorff formula}\label{ME}
\subsubsection*{Formal derivation of Poincar\'e-Hausdorff and Magnus formulas}
Let $(\C\{X\},\partial)$ be the differential algebra freely generated by $X$ (a formal variable)\footnote{It is, in fact, 
the free algebra $\ncp{\C}{(X^{[k]})_{k\ge0}}$ (with $X^{[0]}=X$) endowed with $\partial(X^{[k]})=X^{[k+1]}$, the construction is similar to what is to be found in \cite{VdP}, but in the noncommutative realm.}.  $X$ be a formal variable and $(\C{X},\partial)$ be 
We define a comultiplication $\Delta$ by asking that all $X^{[k]}$ be primitive note that $\Delta$ commutes with the derivation. Setting, in $\widehat{\C\{X\}}$,  $D=\partial(e^X)e^{-X}$, direct computation shows that $D$ is primitive and hence a Lie series, which can therefore be written as a sum of Dynkin trees.\\
On the other hand, the formula
\begin{equation}\label{formal1}
D=\sum_{k\ge1}\dfrac{1}{k!}\sum_{l=0}^{k-1}X^l(\partial X)X^{k-1-l}\cdot \sum_{n\ge0}\dfrac{(-X)^n}{n!}
\end{equation} 
suggests that all bedegrees (in $X,\partial X$) are of the form $[n,1]$ and thus, there exists an univariate series 
$\Phi(Y)=\sum_{n\ge0}a_nY^n$ such that $D=\Phi(ad_X)[\partial X]$.
Using left and right multiplications by $X$ (resp. noted $g,d$), we can rewrite \eqref{formal1} as
\begin{equation}
D=\Big(\sum_{k\ge1}\dfrac{1}{k!}\sum_{l=0}^{k-1}g^ld^{k-1-l}[\partial X]\Big)e^{-X}
\end{equation}
but, from the fact that $g,d$ commute, the inner sum $\sum_{l=0}^{k-1}g^ld^{k-1-l}$ is ruled out by the the following 
identity (in $\C[Y,Z]$, but computed within $\C(Y,Z)$)
\begin{equation}\label{formal2}
\sum_{l=0}^{k-1}Y^lZ^{k-1-l}=\dfrac{Y^k-Z^k}{Y-Z}=\dfrac{\big((Y-Z)+Z\big)^k-Z^k}{Y-Z}=
\sum_{j=1}^{k}\binom{k}{j}(Y-Z)^{j}Z^{k-j}
\end{equation}
Taking notice that $(g-d)=ad_X$ and pluging \eqref{formal2} into \eqref{formal1}, one gets 
\begin{eqnarray}
&&D=\Big(\sum_{k\ge1}\dfrac{1}{k!}\sum_{j=1}^{k}\binom{k}{j}(ad_X)^{j-1}d^{k-j}[\partial X]\Big)e^{-X}=\cr
&&\dfrac{1}{ad_X}\Big(\sum_{k\ge1}\sum_{j=1}^{k}\dfrac{1}{j!(r-j)!}(ad_X)^{j}d^{k-j}[\partial X]\Big)e^{-X}=
\dfrac{e^{ad_X}-1}{ad_X}[X']
\end{eqnarray}
which is Poincar\'e-Hausdorff formula.
\subsubsection*{Application}
Let $G$ be a Lie group with Lie algebra $L$. 
Let $X(t)$ be a $C^1$ path drawn within $L$ (setting as above i.e. $X(0)=0$), then 
\begin{equation}
(e^{X(t)})'e^{-X(t)}=\dfrac{e^{ad_X}-1}{ad_X}[X'(t)]
\end{equation}
In particular, if $S(t)$ is a solution of the system
\begin{equation*}
\Sigma(0,M,e)\qquad \left\{
\begin{array}{ll}
\dfrac{d}{dt}(S(t))=M(t).S(t) &\mbox{   (NCDE)}\cr
S(0)=e & \mbox{    (Init. Cond.)}
\end{array}
\right.
\end{equation*}
then $\Omega=\log(S)$, at a neighbourhood of $t=0$\footnote{Such that the norm of $ad_\Omega$ for the topology of bounded convergence be strictly less that $2\pi$ (the radius of convergence of $\phi(z)=\dfrac{z}{e^{z}-1}$) for which, it is sufficient that 
$\Nm{\Omega}<\pi$.} must satisfy $\Omega'=\dfrac{ad_\Omega}{e^{ad_\Omega}-1}[M]$ (this identity\footnote{The fraction 
$$
\dfrac{ad_X}{e^{ad_X}-1}
$$
means, of course, $\phi(\ad_X)$ where 
$$
\phi(t)=\dfrac{t}{e^{t}-1}=\sum_{n\ge0}B_n\dfrac{t^n}{n!}
$$
the family $(B_n)_{n\ge0}$ being that of Bernouilli numbers.
} 
lives in the completion of the non-commutative free differential algebra generated by the single $X$, constructed like in \cite{VdP} but non 
commutative\footnote{In fact, this (highly noncommutative) differential algebra $(\ncp{\C^{diff}}{X},\partial)$ may be realized as the free algebra 
$\ncp{\C}{(X^{[k]})_{k\ge0}}$ (with $X^{[0]}=X$) endowed with the derivation defined by $\partial(X^{[k]})=X^{[k+1]}$\\ (see e.g. \cite{B_Lie}, Ch I, par. 2.8 \textit{Extension of derivations}).}). 
This guarantees the existence of a local solution drawn on $G$.
\end{document}